\documentclass[11pt,a4paper,UKenglish]{article}
\usepackage{amsmath,amsthm,amsfonts,stmaryrd}
\usepackage[linesnumbered,ruled,noend]{algorithm2e}
\usepackage{hyperref,cleveref}
\title{Notes on Equivalence and Minimization of Weighted Automata}
\author{Stefan Kiefer}
\date{}

\newtheorem{definition}{Definition}[section]
\newtheorem{proposition}[definition]{Proposition}
\newtheorem{theorem}[definition]{Theorem}

\newtheorem{lemma}[definition]{Lemma}

\crefname{equation}{}{}
\crefname{proposition}{Proposition}{Propositions}
\crefname{theorem}{Theorem}{Theorems}
\crefname{corollary}{Corollary}{Corollaries}
\crefname{algorithm}{Algorithm}{Algorithms}
\crefname{section}{Section}{Sections}
\crefname{lemma}{Lemma}{Lemmas}

\begin{document}
%---------------------
\maketitle
\newcommand{\A}{\mathcal{A}}
\newcommand{\AC}{\overline{\mathcal{A}}}
\newcommand{\AF}{\overrightarrow{\mathcal{A}}}
\newcommand{\AB}{\overleftarrow{\mathcal{A}}}
\newcommand{\alphaF}{\overrightarrow{\alpha}}
\newcommand{\Bh}{\widehat{B}}
\newcommand{\BV}{\mathcal{B}}
\newcommand{\br}{\subsection*{Bibliographic Remarks}}
\newcommand{\etaF}{\overrightarrow{\eta}}
\newcommand{\etaB}{\overleftarrow{\eta}}
\newcommand{\FF}{\overrightarrow{F}}
\newcommand{\Fh}{\widehat{F}}
\newcommand{\FN}{F}%{F_{\circ}}
\newcommand{\FV}{\mathcal{F}}
\newcommand{\Ft}{\widetilde{F}}
\newcommand{\K}{\mathbb{K}}
\newcommand{\MC}{\overline{M}}
\newcommand{\MF}{\overrightarrow{M}}
\newcommand{\MB}{\overleftarrow{M}}
\newcommand{\N}{\mathbb{N}}
\newcommand{\NC}{\textup{NC}}
\newcommand{\nF}{\overrightarrow{n}}
\newcommand{\nB}{\overleftarrow{n}}
\renewcommand{\P}[1]{{\cal P}\left(#1\right)}
\newcommand{\Q}{\mathbb{Q}}
\newcommand{\R}{\mathbb{R}}
\newcommand{\rank}{\ensuremath{\textup{rank}}}
\newcommand{\sem}[1]{\llbracket #1 \rrbracket}
\newcommand{\spa}[1]{\langle #1 \rangle}
\newcommand{\X}{\mathcal{X}}
\newcommand{\Y}{\mathcal{Y}}

%\newcommand{\norm}[1]{\lVert#1\rVert}

%\section{Introduction}
\begin{abstract}
This set of notes re-proves known results on weighted automata (over a field, also known as multiplicity automata).
The text offers a unified view on theorems and proofs that have appeared in the literature over decades and were written in different styles and contexts.
\emph{None of the results reported here are claimed to be new.}

The content centres around fundamentals of equivalence and minimization, with an emphasis on algorithmic aspects.
%The selection of the material is based on the author's taste of what he finds fundamental or important or appealing or useful or some combination of these properties.

The presentation is minimalistic.
No attempt has been made to motivate the material.
Weighted automata are viewed from a linear-algebra angle.
As a consequence, the proofs, which are meant to be succinct, but complete and almost self-contained, rely mainly on elementary linear algebra.
\end{abstract}

\section{Preliminaries}

Let $\K$ be a field.
When speaking about algorithms and computational complexity, we will implicitly take as~$\K$ the field~$\Q$ of rational numbers (where we assume that rational numbers are encoded as quotients of integers encoded in binary).
For a finite alphabet~$\Sigma$ we call a map $s : \Sigma^* \to \K$ a \emph{series}.

An \emph{automaton} $\A = (n, \Sigma, M, \alpha, \eta)$ consists of a natural number~$n$ (to which we refer as the \emph{number of states}),
 a finite alphabet~$\Sigma$,
 a map $M : \Sigma \to \K^{n \times n}$, an initial (row) vector $\alpha \in \K^n$, and a final (column) vector $\eta \in \K^n$.
Extend $M$ to a monoid homomorphism $M : \Sigma^* \to \K^{n \times n}$ by setting $M(a_1 \cdots a_k) := M(a_1) \cdots M(a_k)$ and $M(\varepsilon) := I_n$, where $\varepsilon$ is the empty word and $I_n \in \{0,1\}^{n \times n}$ the $n \times n$ identity matrix.
The \emph{semantics} of an automaton~$\A$ is the series $\sem{\A} : \Sigma^* \to \K$ with $\sem{\A}(w) = \alpha M(w) \eta$.
Automata $\A_1, \A_2$ over the same alphabet~$\Sigma$ are said to be \emph{equivalent} if $\sem{\A_1} = \sem{\A_2}$.
An automaton~$\A$ is \emph{minimal} if there is no equivalent automaton~$\A'$ with fewer states.
If $n = 0$, it is natural to put $\sem{\A}(w) = 0$ for all $w \in \Sigma^*$.

We have the following closure properties:
\begin{proposition} \label{prop:closure}
Let $\A_i = (n_i, \Sigma, M_i, \alpha_i, \eta_i)$ for $i \in \{1,2\}$ be automata.
One can compute in logarithmic space (hence, in polynomial time) automata $\A_+, \A_-, \A_\otimes$ with
$\sem{\A_+}(w) = \sem{\A_1}(w) + \sem{\A_2}(w)$ and 
$\sem{\A_-}(w) = \sem{\A_1}(w) - \sem{\A_2}(w)$ and 
$\sem{\A_\otimes}(w) = \sem{\A_1}(w) \cdot \sem{\A_2}(w)$ for all $w \in \Sigma^*$.
One can compute~$\A_+, \A_-$ with $O(|\Sigma| (n_1+n_2)^2)$ arithmetic operations.
One can compute~$\A_\otimes$ with $O(|\Sigma| n_1^2 n_2^2)$ arithmetic operations.
\end{proposition}
\begin{proof}
It is straightforward to check that the automaton $\A_+ = (n_1+n_2, \Sigma, M_+, (\alpha_1, \alpha_2), \eta_+)$ with
\[
M_+(a) \ = \ \begin{pmatrix} M_1(a) & 0_{n_1, n_2} \\ 0_{n_2,n_1} & M_2(a) \end{pmatrix} \quad \text{for all $a \in \Sigma$}
 \quad \text{and} \quad 
 \eta_+ \ = \ \begin{pmatrix} \eta_1 \\ \eta_2 \end{pmatrix}
\]
is the desired automaton, where $0_{m,n}$ stands for the $m \times n$ zero matrix.

The automaton $\A_{-}$ can be constructed similarly to~$\A_+$, but $(\alpha_1, \alpha_2)$ is replaced with $(\alpha_1, -\alpha_2)$.

Let $\mathord{\otimes}$ denote the Kronecker product.
Define $\A_\otimes = (n_1 n_2, \Sigma, M_\otimes, (\alpha_1 \otimes \alpha_2), (\eta_1 \otimes \eta_2))$, where $M_\otimes(a) = M_1(a) \otimes M_2(a)$ for all $a \in \Sigma$.
Using the mixed-product property of~$\mathord{\otimes}$ (i.e., $(A B) \otimes (C D) = (A \otimes C) (B \otimes D)$),
we have for all $a_1 \cdots a_k \in \Sigma^*$:
\begin{align*}
\sem{\A_\otimes}(w) \ &= \ (\alpha_1 \otimes \alpha_2) (M_1(a_1) \otimes M_2(a_1)) \cdots (M_1(a_k) \otimes M_2(a_k))  (\eta_1 \otimes \eta_2) \\
&= \ (\alpha_1 M_1(a_1) \cdots M_1(a_k) \eta_1) \otimes (\alpha_2 M_2(a_1) \cdots M_2(a_k) \eta_2) \\
&= \ \sem{\A_1}(a_1 \cdots a_k) \cdot \sem{\A_2}(a_1 \cdots a_k) \qedhere
\end{align*}
\end{proof}

For a set~$V$ of vectors we use the notation $\spa{v \mid v \in V}$ to denote the vector space spanned by~$V$.
For an automaton~$\A$, define its \emph{forward space} as the (row) vector space $\spa{\alpha M(w) \mid w \in \Sigma^*}$.
Similarly, the \emph{backward space} of~$\A$ is the (column) vector space $\spa{ M(w) \eta \mid w \in \Sigma^* }$.

Let $s : \Sigma^* \to \K$.
The \emph{Hankel matrix} of~$s$ is the (infinite) matrix $H \in \K^{\Sigma^* \times \Sigma^*}$ with $H[x,y] = s(x y)$ for all $x,y \in \Sigma^*$.
Define $\rank(s) := \rank(H)$.

\section{Equivalence Checking} \label{sec:equivalence}

First we discuss how to efficiently compute a basis of the forward space $\FV := \spa{\alpha M(w) \mid w \in \Sigma^*}$ of an automaton $\A = (n, \Sigma, M, \alpha, \eta)$.
It is a matter of basic linear algebra to check that $\FV$ is the smallest vector space that contains~$\alpha$ and is closed under post-multiplication of~$M(a)$ (i.e., $\FV M(a) \subseteq \FV$) for all $a \in \Sigma$.
Hence \cref{alg:forward-space-basic} computes a basis of~$\FV$.

\begin{algorithm}
\DontPrintSemicolon
\If{$\alpha = 0$}{\Return{$\emptyset$}}
$W := \{\varepsilon\}$\;
\While{$\exists\, w \in W\; \exists\, a \in \Sigma : \alpha M(w a) \not\in \spa{\alpha M(w) \mid w \in W}$}{
	$W := W \cup \{w a\}$
   }
\Return{$\{\alpha M(w) \mid w \in W\}$}
\caption{Computing a basis of the forward space of an automaton $\A = (n, \Sigma, M, \alpha, \eta)$.}
\label{alg:forward-space-basic}
\end{algorithm}

\Cref{alg:forward-space-basic} actually computes a set~$W$ of words such that $\{\alpha M(w) \mid w \in W\}$ is a basis of the forward space~$\FV$.
These words will be of interest, e.g., to compute a word~$w$ that ``witnesses'' the inequivalence of two automata.
%We call a set $P \subseteq \Sigma^* \times \K^{n}$ a \emph{word-vector basis} if for all $(w,v) \in P$ we have $v = \alpha M(w)$ and $\{v \mid (w,v) \in P\}$ is a basis of~$\FV$.
Since $\FV$ is a subspace of $\K^{n}$, its dimension, say $\nF$, is at most~$n$.
It follows that $|W| = \nF \le n$ and $|w| \le \nF-1$ holds for all $w \in W$.

We want to make \cref{alg:forward-space-basic} efficient.
First, in addition to the words~$w$ we save the vectors $\alpha M(w)$ to avoid unnecessary vector-matrix computations.
Second, we use a worklist, implemented as a queue, to keep track of which vectors are new in the basis of~$\FV$ computed so far.
% avoid superfluous checks of membership in the subspace of~$\FV$ computed so far.
These refinements result in \cref{alg:forward-space-intermediate}.

\begin{algorithm}
\DontPrintSemicolon
\If{$\alpha = 0$}{\Return{$\emptyset$}}
$P := \{(\varepsilon,\alpha)\}$\;
$Q := [(\varepsilon,\alpha)]$\;
\Repeat{$\mathit{isEmpty}(Q)$}{
 $(w,v) := \mathit{dequeue}(Q)$\;
 \ForAll{$a \in \Sigma$}{
  $w' := w a$\;
  $v' := v M(a)$\;
  \If{$v' \not\in \spa{u \mid (x,u) \in P}$ \label{algline:forward-space-intermediate}}{
   $P := P \cup \{(w', v')\}$\;
   $Q := \mathit{enqueue}(Q,(w',v'))$\;
  }
 }
}
\Return{$P$}
\caption{Computing
$\{(w_1,v_1), \ldots, (w_{\protect\nF}, v_{\protect\nF})\} \subseteq \Sigma^{\le \protect\nF-1} \times \K^{n}$
such that $\{v_1, \ldots, v_{\protect\nF}\}$
is a basis of the forward space of an automaton $\A = (n, \Sigma, M, \alpha, \eta)$}
\label{alg:forward-space-intermediate}
\end{algorithm}

\Cref{algline:forward-space-intermediate} of \Cref{alg:forward-space-intermediate} requires a check for linear independence.
Using Gaussian elimination, such a check can be carried out with $O(n^3)$ arithmetic operations.
To make this more efficient, one can keep a basis of the vector space $\spa{u \mid (x,u) \in P}$ in echelon form.
With such a basis at hand, the check for linear independence amounts to performing one iteration of Gaussian elimination, which takes $O(n^2)$ operations, and checking if the resulting vector is non-zero.
If it is indeed non-zero, it can be added to the basis, thus preserving its echelon form.%
\footnote{For improved numerical stability of the computation, instead of using a basis in echelon form, one may keep an orthonormal basis, against which the new vector is orthogonalized using one iteration ($O(n^2)$ operations) of the modified Gram-Schmidt process.}
 
Since $\nF \le n$, it follows that \cref{algline:forward-space-intermediate} is executed $O(n |\Sigma|)$ times.
Hence we have:
\begin{proposition} \label{prop:compute-forward-intermediate}
Let $\A = (n, \Sigma, M, \alpha, \eta)$ be an automaton.
One can compute in polynomial time (with $O(|\Sigma| n^3)$ arithmetic operations) a set $\{(w_1,v_1), \ldots, (w_{\nF}, v_{\nF})\} \subseteq \Sigma^{\le \nF-1} \times \K^{n}$ such that $\{v_1, \ldots, v_{\nF}\}$ is a basis of~$\FV$ and $v_i = \alpha M(w_i)$ holds for all $1 \le i \le \nF$.
\end{proposition}

An automaton~$\A$ is called \emph{zero} if $\sem{\A}(w) = 0$ for all $w \in \Sigma^*$.
We show:
\begin{proposition} \label{prop:zeroness-in-P}
Let $\A = (n, \Sigma, M, \alpha, \eta)$ be an automaton.
One can check in polynomial time (with $O(|\Sigma| n^3)$ arithmetic operations) whether $\A$ is zero, and if it is not, output $w \in \Sigma^*$ with $|w| \le n-1$ such that $\sem{\A}(w) \ne 0$.
\end{proposition}
\begin{proof}
Automaton~$\A = (n, \Sigma, M, \alpha, \eta)$ is zero if and only if its forward space $\FV := \spa{\alpha M(w) \mid w \in \Sigma^*}$ is orthogonal to $\eta$, i.e., $v \eta = 0$ for all $v \in \FV$.
Let $S \subseteq \K^n$ (with $|S| \le n$) be a basis of~$\FV$.
Then $\A$ is zero if and only if $v \eta = 0$ holds for all $v \in S$.
But by \cref{prop:compute-forward-intermediate} one can compute such~$S$.
Similarly, one can compute, if it exists, the ``counterexample''~$w$.
\end{proof}

\begin{theorem} \label{thm:equivalence-in-P}
Let $\A_i = (n_i, \Sigma, M_i, \alpha_i, \eta_i)$ for $i \in \{1,2\}$ be automata.
One can check in polynomial time (with $O(|\Sigma| (n_1+n_2)^3)$ arithmetic operations) whether $\A_1, \A_2$ are equivalent, and if they are not, output $w \in \Sigma^*$ with $|w| \le n_1 + n_2 - 1$ such that $\sem{\A_1}(w) \ne \sem{\A_2}(w)$.
\end{theorem}
\begin{proof}
Compute automaton~$\A_-$ from \cref{prop:closure}.
Then the theorem follows from \cref{prop:zeroness-in-P}.
\end{proof}

\br

Equivalence checking goes back to the seminal paper by Sch\"utzenberger from 1961~\cite{IC::Schutzenberger1961}. % and his ``standardisation'' algorithm.
A polynomial-time algorithm could be derived from there but was not made explicit.
The books by Paz~\cite{Paz71} and Eilenberg~\cite{Eilenberg74} from 1971 and 1974, respectively, describe an exponential-time algorithm based on the fact that shortest ``counterexamples'' have length at most $n_1+n_2-1$.
An $O(|\Sigma| (n_1 + n_2)^4)$ (in terms of arithmetic operations) algorithm was explicitly provided in 1992 by Tzeng~\cite{Tzeng92}.
Improvements to $O(|\Sigma| (n_1 + n_2)^3)$ were then \mbox{(re-)}discovered, e.g., in \cite{CortesMohriRastogi07,KieferMOWW11,BerlinkovFS18}. These improvements are all based on the idea described before \cref{prop:compute-forward-intermediate}.
The abstract of the 2002 paper~\cite{Archangelsky02} indicates that this improvement was already known to some.
Incidentally, a different algorithm, also cubic in~$n$, was proposed in~\cite{Archangelsky02}.

\section{Minimization} \label{sec:minimization}

Let $\A = (n, \Sigma, M, \alpha, \eta)$ be an automaton.
Let $F \in \K^{\nF \times n}$ with $\nF \le n$ be a matrix whose rows form a basis of the forward space~$\FV$.
Similarly, let $B \in \K^{n \times \nB}$ with $\nB \le n$ be a matrix whose columns form a basis of the backward space~$\BV$.
Since $\FV M(a) \subseteq \FV$ and $M(a) \BV \subseteq \BV$ for all $a \in \Sigma$,
 there exist maps $\MF : \Sigma \to \K^{\nF \times \nF}$ and $\MB : \Sigma \to \K^{\nB \times \nB}$ such that
 \begin{equation*}
  F M(a) \ = \ \MF(a) F \quad \text{and} \quad M(a) B \ = \ B \MB(a) \quad \text{for all $a \in \Sigma$.}
 \end{equation*}
These maps $\MF, \MB$ are unique, as $F, B$ have full rank.
The above equalities extend inductively to words:
 \begin{equation}
  F M(w) \ = \ \MF(w) F \quad \text{and} \quad M(w) B \ = \ B \MB(w) \quad \text{for all $w \in \Sigma^*$}
   \label{eq:commutativity}
 \end{equation}
Let $\alphaF \in \K^{\nF}$ be the unique row vector with $\alphaF F = \alpha$, and $\etaB \in \K^{\nB}$ be the unique column vector with $B \etaB = \eta$.
Call $\AF := (\nF, \Sigma, \MF, \alphaF, F \eta)$ the \emph{forward conjugate} of~$\A$ with base~$F$,
 and $\AB := (\nB, \Sigma, \MB, \alpha B, \etaB)$ the \emph{backward conjugate} of~$\A$ with base~$B$.

\begin{proposition} \label{prop:equivalence}
 Let $\A$ be an automaton. Then $\sem{\A} = \sem{\AF} = \sem{\AB}$.
\end{proposition}
\begin{proof}
 By symmetry, it suffices to show the first equality.
 Indeed, we have for all $w \in \Sigma^*$:
 \begin{align*}
   \sem{\AF}(w)
    \ & =\ \alphaF \MF(w) F \eta \\
    \ & =\ \alphaF F M(w) \eta && \text{by~\cref{eq:commutativity}} \\
    \ & =\ \alpha M(w) \eta && \text{definition of~$\alphaF$} \\
    \ & =\ \sem\A(w) \tag*{\qedhere}
 \end{align*}
\end{proof}

\begin{proposition} \label{prop:compute-conjugate}
Let $\A = (n, \Sigma, M, \alpha, \eta)$ be an automaton.
One can compute in polynomial time (with $O(|\Sigma| n^3)$ arithmetic operations)
\begin{itemize}
\item a matrix~$F$ whose rows form a basis of the forward space of~$\A$, and
\item the forward conjugate of~$\A$ with base~$F$.
\end{itemize}
The statement holds analogously for backward conjugates.
\end{proposition}
\begin{proof}
By \cref{prop:compute-forward-intermediate} one can compute a basis of the forward space, and hence~$F$, in the required time.
Having computed~$F$ it is straightforward to compute $\AF$ in the required time.
The same holds analogously for~$\AB$.
\end{proof}

%The following propositions are known (see Carlyle/Paz'71, Fliess'74, Beimel et al.'00).
\begin{proposition} \label{prop:direction-1}
 Let $\A$ be an automaton. Then $\rank(\sem\A) \le n$.
\end{proposition}
\begin{proof}
 Consider the matrices $\Fh: \K^{\Sigma^* \times n}$ and $\Bh: \K^{n \times \Sigma^*}$ with
  $\Fh[w,\cdot] = \alpha M(w)$ and $\Bh[\cdot,w] = M(w) \eta$ for all $w \in \Sigma^*$.
 Note that $\rank(\Fh) \le n$ (and similarly $\rank(\Bh) \le n$).
 Let $x,y \in \Sigma^*$.
 Then $(\Fh \Bh)[x,y] = \alpha M(x) M(y) \eta = \sem\A(x y)$, so $\Fh \Bh$ is the Hankel matrix of~$\sem\A$.
 Hence $\rank(\sem\A) = \rank(\Fh \Bh) \le \rank(\Fh) \le n$.
\end{proof}

Call an automaton with $n$ states \emph{forward-minimal} (resp., \emph{backward-minimal}) if its forward (resp., backward) space has dimension~$n$.

\begin{proposition} \label{prop:forw-conj-is-forw-min}
A forward conjugate is forward-minimal.
A backward conjugate is backward-minimal.
\end{proposition}
\begin{proof}
By symmetry, it suffices to prove the statement about forward conjugates.
Let $\AF = (\nF, \Sigma, \MF, \alphaF, F \eta)$ be the forward conjugate of $\A = (n, \Sigma, M, \alpha, \eta)$ with base~$F \in \K^{\nF \times n}$.
We have:
\begin{align*}
   & \dim\, \spa{\alphaF \MF(w) \mid w \in \Sigma^*} \\
=\ & \dim\, \spa{\alphaF \MF(w) F \mid w \in \Sigma^*} && \text{the rows of~$F$ are linearly independent}\\
=\ & \dim\, \spa{\alphaF F M(w) \mid w \in \Sigma^*} && \text{by \cref{eq:commutativity}}\\
=\ & \dim\, \spa{\alpha M(w) \mid w \in \Sigma^*} && \text{definition of~$\alphaF$}\\
=\ & \dim\, \FV && \text{definition of~$\FV$} \\
=\ & \nF && \text{definition of~$\nF$} &&\qedhere
\end{align*}
\end{proof}

\begin{proposition} \label{prop:new-minimal}
A backward conjugate of a forward-minimal automaton is minimal.
A forward conjugate of a backward-minimal automaton is minimal.
\end{proposition}
\begin{proof}
By symmetry, it suffices to show the first statement.
Let $\A = (n, \Sigma, M, \alpha, \eta)$ be forward-minimal.
Let $B \in \K^{n \times \nB}$ be a matrix whose columns form a basis of the backward space of~$\A$.
By \cref{prop:direction-1} it suffices to show that $\nB = \rank(H)$, where $H$ is the Hankel matrix of~$\sem{\A}$.
Let $\Fh$ and $\Bh$ be the matrices from the proof of \cref{prop:direction-1}.
Since $\A$ is forward-minimal, the columns of~$\Fh$ are linearly independent.
We have:
\begin{align*}
&\ \nB \\ 
= &\ \rank(B) && \text{definition of~$B$} \\
= &\ \rank(\Fh B) && \text{the columns of~$\Fh$ are linearly independent} \\
= &\, \dim\, \spa{\Fh M(w) \eta \mid w \in \Sigma^*} && \text{definition of~$B$} \\
= &\ \rank(\Fh \Bh) && \text{definition of~$\Bh$} \\
= &\ \rank(H) && \text{proof of \cref{prop:direction-1}} \qedhere
\end{align*}
\end{proof}

\begin{theorem} \label{thm:minimization}
Let $\A$ be an automaton.
Let $\A'$ be a backward conjugate of a forward conjugate of~$\A$ (or a forward conjugate of a backward conjugate of~$\A$).
Then $\A'$ is minimal and equivalent to~$\A$.
It can be computed in polynomial time (with $O(|\Sigma| n^3)$ arithmetic operations).
\end{theorem}
\begin{proof}
Minimality follows from \cref{prop:forw-conj-is-forw-min,prop:new-minimal}.
Equivalence follows from \cref{prop:equivalence}.
Polynomial-time computability follows by invoking \cref{prop:compute-conjugate} twice.
\end{proof}

Let $\A_i = (n, \Sigma, M_i, \alpha_i, \eta_i)$ for $i \in \{1,2\}$ be minimal, where $\A_2$ is the forward conjugate of~$\A_1$ with some base $Q \in \K^{n \times n}$.
By minimality and \cref{prop:equivalence}, matrix~$Q$ is invertible.
Since
\[
 \alpha_2 Q = \alpha_1, \quad \eta_2 = Q \eta_1, \quad Q M_1(a) = M_2(a) Q \quad\text{for all } a \in \Sigma\,,
\]
automaton~$\A_1$ is the backward conjugate of~$\A_2$ with base~$Q$.
Since
\[
\alpha_2 = \alpha_1 Q^{-1}, \quad Q^{-1} \eta_2 = \eta_1, \quad M_1(a) Q^{-1} = Q^{-1} M_2(a) \quad\text{for all } a \in \Sigma\,,
\]
automaton~$\A_1$ is the forward conjugate of~$\A_2$ with base~$Q^{-1}$, and $\A_2$ is the backward conjugate of~$\A_1$ with base~$Q^{-1}$.

This motivates the following definition.
Call minimal automata $\A_1, \A_2$ \emph{conjugate} if one is a forward (equivalently, backward) conjugate of the other.

\begin{theorem} \label{thm:conjugate-new}
Two minimal automata are conjugate if and only if they are equivalent.
\end{theorem}
\begin{proof}
The forward direction follows from \cref{prop:equivalence}.

Towards the backward direction, let $\A_i = (n, \Sigma, M_i, \alpha_i, \eta_i)$ for $i \in \{1,2\}$ be minimal equivalent automata.
For $i \in \{1, 2\}$, consider the matrices $\Fh_i: \K^{\Sigma^* \times n}$ and $\Bh_i: \K^{n \times \Sigma^*}$ with $\Fh_i[w,\cdot] = \alpha_i M_i(w)$ and $\Bh_i[\cdot,w] = M_i(w) \eta_i$ for all $w \in \Sigma^*$.
It follows from minimality and \cref{prop:equivalence} that $\Fh_i$ and~$\Bh_i$ have full rank~$n$.
Thus, there is an invertible matrix~$Q \in \K^{n \times n}$ with $\Fh_1 = \Fh_2 Q$.
We show that $\A_2$ is the forward conjugate of~$\A_1$ with base~$Q$.

We have $\alpha_1 = \Fh_1[\varepsilon,\cdot] = \Fh_2[\varepsilon,\cdot] Q = \alpha_2 Q$.
Letting $H$ denote the Hankel matrix of $\sem{\A_1} = \sem{\A_2}$, we have $\Fh_2 Q \Bh_1 = \Fh_1 \Bh_1 = H = \Fh_2 \Bh_2$.
Since $\Fh_2$ has full rank, it follows that $Q \Bh_1 = \Bh_2$.
In particular, $Q \eta_1 = Q \Bh_1[\cdot, \varepsilon] = \Bh_2[\cdot, \varepsilon] = \eta_2$.

For any $a \in \Sigma$ let $H_a \in \K^{\Sigma^* \times \Sigma^*}$ be the matrix with $H_a[x,y] = \sem{\A_i}(x a y) $ for all $x, y \in \Sigma^*$.
We have $\Fh_i M_i(a) \Bh_i = H_a$ for all $a \in \Sigma$.
Thus, for all $a \in \Sigma$ we have:
\[ \Fh_2 Q M_1(a) \Bh_1 \ = \ \Fh_1 M_1(a) \Bh_1 \ = \ H_a \ = \ \Fh_2 M_2(a) \Bh_2 \ = \ \Fh_2 M_2(a) Q \Bh_1
\]
Since $\Fh_2$ and~$\Bh_1$ have full rank, we have $Q M_1(a) = M_2(a) Q$ for all $a \in \Sigma$.
\end{proof}

\br

Minimization is closely related to equivalence and also goes back to~\cite{IC::Schutzenberger1961}.
The book~\cite[Chapter~II]{BerstelReutenauer88} describes a minimization procedure.
An $O(|\Sigma| n^4)$ minimization algorithm (for a related probabilistic model) was suggested in~\cite{GillmanS94}.
The algorithm given in this note reminds of Brzozowski's algorithm for minimizing DFAs~\cite{Brzozowski62}.
The succinct formulation in this note is essentially from~\cite{14KW-ICALP}.
Further generalizations of Brzozowski's algorithm are discussed in~\cite{BonchiBrzozowski}.

\Cref{thm:conjugate-new} also goes back to \cite{IC::Schutzenberger1961}.
See also \cite{fliess} and \cite[Chapter~II]{BerstelReutenauer88}.

\section{The Hankel Automaton} \label{sec-Hankel-aut}

Let $s : \Sigma^* \to \K$ be a series of rank~$n$, with Hankel matrix~$H$.
Call a set $C = \{c_1, \ldots, c_n\} \subseteq \Sigma^*$ \emph{complete} if the columns of~$H[\cdot, C]$ form a basis of the column space of~$H$.

Note that for any $G \subseteq \Sigma^*$ and $w \in \Sigma^*$ we have $H[G w, C] = H[G, w C]$, where $G w := \{g w \mid g \in G\}$ and $w C := \{w c \mid c \in C\}$.

Let $C = \{c_1, \ldots, c_n\} \subseteq \Sigma^*$ be complete.
Then, for any $w \in \Sigma^*$ there is a unique column vector $\eta_w \in \K^{n}$ with $H[\cdot,w] = H[\cdot,C] \eta_w$, and for all $a \in \Sigma$ a unique matrix~$\MC(a) \in \K^{n \times n}$ with $H[\cdot,C] \MC(a) = H[\cdot, a C]$.
We define the \emph{Hankel automaton} for $s,C$ as $\AC = (n, \Sigma, \MC, H[\varepsilon,C], \eta_\varepsilon)$.

\begin{proposition} \label{prop:canonical-aut-monoid}
Let $\AC = (n, \Sigma, \MC, H[\varepsilon,C], \eta_\varepsilon)$ be the Hankel automaton for~$s,C$.
Then for all $w \in \Sigma^*$ we have $H[\cdot,C] \MC(w) =  H[\cdot, w C]$.
Hence, if $G = \{g_1, \ldots, g_n\} \subseteq \Sigma^*$ is such that $H[G,C]$ has full rank, we have $\MC(w) = H[G,C]^{-1} H[G, w C]$ for all $w \in \Sigma^*$.
\end{proposition}
\begin{proof}
We proceed by induction on the length of~$w$.
The induction base ($w = \varepsilon$) is trivial.
For the step, let $w \in \Sigma^*$ and $a \in \Sigma$.
We have:
\begin{align*}
H[\cdot,C] \MC(w a)
\ &=\ H[\cdot,w C] \MC(a) && \text{induction hypothesis} \\
\ &=\ H[\cdot\, w, C] \MC(a) \\
\ &=\ H[\cdot\, w, a C] && \text{definition of~$\MC(a)$} \\
\ &=\ H[\cdot, w a C] && \qedhere
\end{align*}
\end{proof}

\begin{proposition} \label{prop:canonical-main}
Let $\AC = (n, \Sigma, \MC, H[\varepsilon,C], \eta_\varepsilon)$ be the Hankel automaton for~$s,C$.
Then $\sem{\AC} = s$ and $\AC$ is minimal.
\end{proposition}
\begin{proof}
We have for all $w \in \Sigma^*$:
\begin{align*}
\sem{\AC}(w)
\ &=\ H[\varepsilon,C] \MC(w) \eta_\varepsilon \\
\ &=\ H[\varepsilon, w C] \eta_\varepsilon && \text{\cref{prop:canonical-aut-monoid}} \\
\ &=\ H[w,C] \eta_\varepsilon \\
\ &=\ H[w,\varepsilon] && \text{definition of~$\eta_\varepsilon$} \\
\ &=\ s(w)
\end{align*}
Minimality follows from \cref{prop:direction-1}.
\end{proof}

\begin{theorem} \label{thm:WA-are-complete}
Let $s : \Sigma^* \to \K$ be a series and $n \in \N$.
Then $\rank(s) \le n$ if and only if there is an automaton~$\A$ with $\sem{\A} = s$ that has at most $n$ states.
\end{theorem}
\begin{proof}
Follows from \cref{prop:canonical-main,{prop:direction-1}}.
\end{proof}

%\begin{lemma} \label{lem:backward-Hankel}
%Let $\A = (n, \Sigma, M, \alpha, \eta)$ be an automaton whose forward space has dimension~$n$.
%Let $H$ be the Hankel matrix of~$\sem{\A}$.
%Let $\Bh: \K^{n \times \Sigma^*}$ be the matrix with $\Bh[\cdot,w] = M(w) \eta$ for all $w \in \Sigma^*$.
%Then, for any $C \subseteq \Sigma^*$, the columns of $H[\cdot,C]$ form a basis of the column space of~$H$ if and only if the columns of $\Bh[\cdot,C]$ form a basis of the backward space of~$\A$.
%\end{lemma}
%\begin{proof}
%Let $\Fh: \K^{\Sigma^* \times n}$ be the matrix with $\Fh[w,\cdot] = \alpha M(w)$ for all $w \in \Sigma^*$.
%We have $\Fh \Bh = H$.
%Since the rows of~$\Fh$ span the forward space of~$\A$, matrix~$\Fh$ has full rank~$n$, and so the columns of~$\Fh$ are linearly independent.
%The lemma follows.
%\end{proof}

The following proposition uses some notions from \cref{sec:minimization}.

\begin{proposition} \label{prop:canonical-conjugate}
Let $\A = (n, \Sigma, M, \alpha, \eta)$ be forward-minimal.
Let $C = \{c_1, \ldots, c_r\} \subseteq \Sigma^*$ be such that the columns of the matrix $B := (M(c_1) \eta, \ldots, M(c_r) \eta) \in \K^{n \times r}$ form a basis of the backward space of~$\A$.
Then $C$ is complete, and the backward conjugate of~$\A$ with base~$B$ is the Hankel automaton for~$\sem{\A},C$.
\end{proposition}
\begin{proof}
Let $H$ be the Hankel matrix of~$\sem{\A}$.
Let $\Fh: \K^{\Sigma^* \times n}$ and $\Bh: \K^{n \times \Sigma^*}$ be the matrices with $\Fh[w,\cdot] = \alpha M(w)$ and $\Bh[\cdot,w] = M(w) \eta$ for all $w \in \Sigma^*$.
We have $\Fh \Bh = H$ and $\Bh[\cdot,C] = B$, hence $\Fh B = H[\cdot,C]$.
Since the column spaces of $\Bh$ and~$B$ are equal, it follows that the column spaces of $H$ and~$H[\cdot,C]$ are equal.
Since $\A$ is forward-minimal, $\Fh$ has full rank.
Thus, $r = \rank(B) = \rank(\Fh B) = \rank(H[\cdot,C])$, so the columns of $H[\cdot,C]$ are linearly independent.
Hence, $C$ is complete.

Let $\AC = (r, \Sigma, \MC, H[\varepsilon,C], \eta_\varepsilon)$ be the Hankel automaton for~$\sem{\A},C$.
For all $x, w \in \Sigma^*$ and all $c \in C$ we have $\Fh[x,\cdot] M(w) B[\cdot,c] = \alpha M(x) M(w) M(c) \eta = \sem{\A}(x w c) = H[x, w c]$.  
Thus, we have for all $w \in \Sigma^*$:
\begin{align*}
 \Fh M(w) B
 \ &=\ H[\cdot, w C] \\
 \ &=\ H[\cdot, C] \MC(w) && \text{\cref{prop:canonical-aut-monoid}} \\
 \ &=\ \Fh B \MC(w)
\end{align*}
Since $\Fh$ has full rank, it follows that $M(w) B = B \MC(w)$ for all $w \in \Sigma^*$.
Similarly, we have $\Fh B \eta_\varepsilon = H[\cdot, C] \eta_\varepsilon = H[\cdot, \varepsilon] = \Fh \eta$, and since $\Fh$ has full rank, it follows that $B \eta_\varepsilon = \eta$.
Finally, we have $\alpha B = \Fh[\varepsilon,\cdot] B = H[\varepsilon,C]$.
We conclude that $\AC$ is the backward conjugate of~$\A$ with base~$B$.
\end{proof}

\begin{theorem} \label{thm:compute-conjugate-Hankel}
Let $\A = (n, \Sigma, M, \alpha, \eta)$ be an automaton.
One can compute in polynomial time (with $O(|\Sigma| n^3)$ arithmetic operations) a complete (for~$\sem{\A}$) set $C = \{c_1, \ldots, c_r\} \subseteq \Sigma^{\le r-1}$ with $r \le n$, and the Hankel automaton for $\sem{\A}, C$.
\end{theorem}
\begin{proof}
Using \cref{prop:compute-conjugate,prop:equivalence,prop:forw-conj-is-forw-min}, first compute in the required time a forward-minimal automaton~$\AF = (\nF, \Sigma, \MF, \alphaF, \etaF)$ with $\sem{\AF} = \sem{\A}$.
By the backward analogue of \cref{prop:compute-forward-intermediate} one can compute in the required time a set $C = \{c_1, \ldots, c_r\} \subseteq \Sigma^{\le r-1}$ and the matrix $B = (\MF(c_1) \etaF, \ldots, \MF(c_r) \etaF) \in \K^{\nF \times r}$ such that the columns of~$B$ form a basis of the backward space of~$\AF$.
Let $\A'$ be the backward conjugate of~$\AF$ with base~$B$.
By \cref{prop:compute-conjugate}, it can be computed in the required time.
By \cref{prop:canonical-conjugate}, $\A'$ is the Hankel automaton for $\sem{\A}, C$.
\end{proof}

\br

The material in this section, at least up to \cref{thm:WA-are-complete}, is similar to \cite[Section~2]{carlyle1971realizations} and \cite{fliess}.
See also \cite[Theorem~5.3]{MandelSimon77} and \cite[Section~2]{five} for related treatments.

\section{Computations in \NC}

We show that some of the mentioned polynomial-time computations can even be carried out in the complexity class~\NC, which comprises those languages having L-uniform Boolean circuits of polylogarithmic depth and polynomial size, or, equivalently,
those problems solvable in polylogarithmic time on parallel random-access machines with polynomially many processors. We have $\text{NL} \subseteq \text{\NC} \subseteq \text{P}$.

\begin{lemma} \label{lem: AtA}
Let $A \in \K^{m \times n}$.
The row spaces of $A$ and $A^T A$ are equal.
\end{lemma}
\begin{proof}
It is clear that the row space of $A^T A$ is included in the row space of~$A$.
For the converse, it suffices to show that the null space of $A^T A$ is included in the null space of~$A$.
Let $x \in \K^n$ with $A^T A x = 0_n$, where $0_n$ denotes the zero vector.
Then $(A x)^T (A x) = x^T A^T A x = x^T 0_n = 0$, and hence $A x = 0_m$.
\end{proof}

In the following we assume $\K = \Q$.

\begin{proposition}[cf.~\cref{prop:compute-forward-intermediate}] \label{prop:compute-forward-intermediate-NC}
Let $\A = (n, \Sigma, M, \alpha, \eta)$ be an automaton.
One can compute in \NC\ a basis of the forward space $\FV := \spa{\alpha M(w) \mid w \in \Sigma^*}$.
\end{proposition}
\begin{proof}
Let $\FN \in \K^{\Sigma^{\le n-1} \times n}$ be the matrix with $\FN[w,\cdot] = \alpha M(w)$ for $w \in \Sigma^{\le n-1}$.
We have shown in \cref{sec:equivalence} that the rows of~$\FN$ span~$\FV$.
By \cref{lem: AtA}, the rows of~$E := \FN^T \FN \in \K^{n \times n}$ also span~$\FV$.

Let $e_i \in \{0,1\}^n$ denote the $i$th coordinate column vector, and let $\mathord{\otimes}$ denote the Kronecker product.
Using the mixed-product property of~$\mathord{\otimes}$ (i.e., $(A B) \otimes (C D) = (A \otimes C) (B \otimes D)$), we have:
\begin{align*}
E[i,j]
\ &=\ \FN^T[i,\cdot] \ \FN[\cdot,j] \\
\ &=\ \sum_{w \in \Sigma^{\le n-1}} (\alpha M(w))[i] \ (\alpha M(w))[j] \\
\ &=\ \sum_{w \in \Sigma^{\le n-1}} (\alpha M(w) e_i) \otimes (\alpha M(w) e_j)  \\
\ &=\ \sum_{w \in \Sigma^{\le n-1}} (\alpha \otimes \alpha) (M(w) \otimes M(w)) (e_i \otimes e_j) \\
\ &=\ (\alpha \otimes \alpha) \left(\sum_{w \in \Sigma^{\le n-1}} M(w) \otimes M(w)\right) (e_i \otimes e_j) \\
\ &=\ (\alpha \otimes \alpha) \left(\sum_{k=0}^{n-1} \sum_{w \in \Sigma^k} M(w) \otimes M(w)\right) (e_i \otimes e_j) \\
\ &=\ (\alpha \otimes \alpha) \left(\sum_{k=0}^{n-1} \sum_{a_1 \cdots a_k \in \Sigma^k} (M(a_1) \otimes M(a_1)) \cdots (M(a_k) \otimes M(a_k)) \right) \\ & \qquad (e_i \otimes e_j) \\
\ &=\ (\alpha \otimes \alpha) \left(\sum_{k=0}^{n-1} \Big(\sum_{a \in \Sigma} M(a) \otimes M(a)\Big)^k\right) (e_i \otimes e_j)
\end{align*}
Since Kronecker products, sums and matrix powers~\cite{Cook85} can be computed in~\NC, one can compute~$E$ in~\NC.
We include the $i$th row of~$E$ in the desired basis of~$\FV$ if and only if $\rank(E[\{1, \ldots, i\},\cdot]) > \rank(E[\{1, \ldots, i-1\},\cdot])$.
This can be done in~\NC, as the rank of a matrix can be determined in \NC~\cite{IbarraMR80}.
\end{proof}

\begin{proposition}[cf.~\cref{prop:zeroness-in-P}] \label{prop:zeroness-in-P-NC}
Let $\A = (n, \Sigma, M, \alpha, \eta)$ be an automaton.
One can check in \NC\ whether $\A$ is zero.
\end{proposition}
\begin{proof}
Analogous to the proof of \cref{prop:zeroness-in-P}.
\end{proof}

\begin{theorem}[cf.~\cref{thm:equivalence-in-P}] \label{thm:equivalence-in-P-NC}
Let $\A_i = (n_i, \Sigma, M_i, \alpha_i, \eta_i)$ for $i \in \{1,2\}$ be automata.
One can check in \NC\ whether $\A_1, \A_2$ are equivalent.
\end{theorem}
\begin{proof}
Analogous to the proof of \cref{thm:equivalence-in-P}.
\end{proof}

\begin{proposition}[cf.~\cref{prop:compute-conjugate}] \label{prop:compute-conjugate-NC}
Let $\A = (n, \Sigma, M, \alpha, \eta)$ be an automaton.
One can compute in~\NC:
\begin{itemize}
\item a matrix~$F$ whose rows form a basis of the forward space of~$\A$,
\item the forward conjugate of~$\A$ with base~$F$.
\end{itemize}
The statement holds analogously for backward conjugates.
\end{proposition}
\begin{proof}
The first item follows from \cref{prop:compute-forward-intermediate-NC}.
The second item follows from the fact that linear systems of equations can be solved in \NC~\cite{Csanky76}.
\end{proof}

\begin{theorem}[cf.~\cref{thm:minimization}] \label{thm:minimization-NC}
Given an automaton, one can compute in~\NC\ a minimal equivalent automaton.
\end{theorem}
\begin{proof}
Analogous to the proof of \cref{thm:minimization}.
\end{proof}

\br

\Cref{thm:equivalence-in-P-NC} about equivalence checking was proved by Tzeng~\cite{Tzeng96}.
\Cref{thm:minimization-NC} about minimization was obtained in~\cite[Section~4.2]{17KMW-LMCS}.
It is not known whether ``counterexample'' words for equivalence can be computed in~\NC.
They can be computed in randomized~\NC~\cite{KieferMOWW13}.

\paragraph*{Acknowledgements.}
The author thanks Oscar Darwin, Qiyi Tang, and Cas Widdershoven for comments that helped to improve the text.

\bibliographystyle{plain} %oder alpha oder splncs
\bibliography{references}
\end{document}